%% file: arxiv.tex
\documentclass{amsart}
\usepackage[pdftex]{graphicx}

\usepackage{bbm}

\usepackage[latin1]{inputenc}

\usepackage{epsfig}
\usepackage{color}

\newtheorem{proposition}{Proposition}
\newcommand{\argmax}{\operatornamewithlimits{argmax}}

\newtheorem{corollary}{Corollary}

\def\N{{\mathbb N}}

\def\P{{\mathbb P}}
\def\E{{\mathbb E}}

\def\cal{\mathcal}
\def\eps{\varepsilon}
\def\etal{{\em et al. }}

\newcommand\ind[1]{\mathbbm{1}_{\left\{#1\right\}}}

\title{Traffic Capacity of Large WDM Passive Optical Networks}

\author[N. Antunes]{Nelson Antunes}
\address[N. Antunes]{INRIA Paris--Rocquencourt, France on leave of absence of University of the Algarve, Faro, Portugal}
\email{NAntunes@ualg.pt}

\author[C. Fricker]{Christine Fricker}
\address[Ch. Fricker, Ph. Robert, J. Roberts]{INRIA Paris --- Rocquencourt,  Domaine de Voluceau, 78153 Le Chesnay, France.}
\email{Christine.Fricker@inria.fr}

\author[Ph. Robert]{Philippe  Robert}
\email{Philippe.Robert@inria.fr}
\urladdr{http://www-rocq.inria.fr/\string~robert}

\author[J. Roberts]{James  Roberts}
\email{James.Roberts@inria.fr}

\begin{document}

\begin{abstract}
As passive optical networks (PON) are increasingly deployed to provide high speed Internet
access, it is important to understand their fundamental traffic capacity limits. The paper
discusses performance  models applicable to  wavelength division multiplexing  (WDM) EPONs
and  GPONs under the  assumption that  users access  the fibre  via optical  network units
equipped  with  tunable  transmitters.  The  considered stochastic  models  are  based  on
multiserver polling systems  for which explicit analytical results are  not known. A large
system asymptotic,  mean-field approximation, is used  to derive closed  form solutions of
these complex systems. Convergence  of the mean field dynamics is proved  in the case of a
simple network  configuration.  Simulation results show  that, for a  realistic sized PON,
the mean field approximation is accurate.
\end{abstract}

\maketitle

\section{Introduction}
Figure \ref{fig:ponpicture} illustrates the main components of a passive optical network
(PON). Optical network units (ONUs) are situated in or close to user premises. An ONU
might be dedicated to an individual user or, when fiber is terminated at the curb or at
the building, it could concentrate the traffic of several users. The ONUs  are controlled
from an optical line termination (OLT) equipment that realizes the interface with the wide
area network. Passive splitters are used to broadcast downstream optical signals to the
ONUs and to merge upstream signals destined to the OLT.  

To  increase  the  capacity of  a  PON,  it  is  proposed  to employ  wavelength  division
multiplexing (WDM) to multiply the capacity of the physical fibre. Typically the ONUs will
not be equipped to use all wavelengths  simultaneously. In this paper we assume an ONU has
one or more transmitters that can be  tuned to any wavelength. The spectrum is thus shared
efficiently  but   the  capacity  of   an  ONU  is   limited  by  its  quota   of  tunable
transmitters\footnote{Equivalent capacity constraints  would arise if the PON  were to use
one or more reflective semi-conductor optical amplifiers in the ONUs~\cite{Mai09}}.

\begin{figure}[htp]
\centering
\scalebox{0.7}{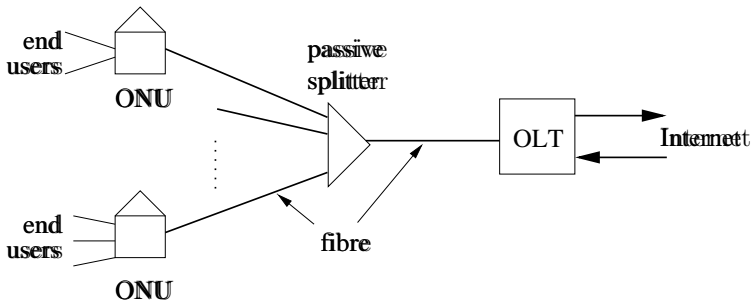}
\caption{PON components: ONUs, passive splitters and OLT}\label{fig:ponpicture}
\end{figure}

The major PON design concern is the dynamic bandwidth allocation (DBA) algorithm
implemented by the OLT to arbitrate use of upstream capacity and to avoid
collisions. There are two PON standards, EPON from IEEE and GPON from ITU, that differ in
the way they manage resource sharing. EPON relies on asynchronous transmission of Ethernet
packets while GPON implements a synchronous 125$\mu$s frame enabling both circuit and
packet switching. We develop generic models for packet based services that are applicable
for both standards.  

The focus of the paper is on traffic capacity, the limit on demand beyond which some ONU
queues would always be saturated. This is arguably the most critical performance for the
PON whose very high bit rate ensures delays are tiny at normal loads and only become
significant as demand attains the capacity limit (see~\cite{AFRR10,MRC07}).

We assume the OLT  manages a single queue per user, receiving  queue size \textit{reports} from the
ONUs and attributing corresponding {\em grants}  for upstream transmission. A grant is the
time interval during which a user may transmit. We maintain that such simple management at
PON level  (i.e., no accounting for intra-ONU class of
service differentiation at the OLT)  is highly  desirable to  limit complexity  and generally  meets user
requirements, the  OLT ensures users receive a fair share of bandwidth,  the ONUs realize
what service differentiation is needed within that share by choosing which packets to send
to fulfill the allocated grants.

The notion of fair allocation naturally leads to a form of polling. For EPON we define a
DBA that realizes a classical polling system with a particular interpretation for the
switchover time between queues. The same model can be applied to one possible DBA
implementation for GPON.  A second considered GPON DBA, on the other hand, is more
naturally modelled as a head of line processor sharing system. See below for more detail. 

The underlying polling systems are complicated by the fact that ONU capacity is limited by
the number  of tunable transmitters  they have,  i.e., the number  of servers that  can be
present simultaneously at an ONU. In particular, stability conditions of multiserver polling
systems with limited-gated policies appear intractable in general.


\paragraph*{Related work}
There  is  an  abundant  literature  on   PONs  and  numerous  DBA  algorithms  have  been
proposed.  Two   recent  surveys,   by  McGarry~\etal  \cite{MRM08}   and  by   Zheng  and
Mouftah~\cite{ZM09}, present  a comprehensive review  for EPON. Skubic \etal~\cite{SCA09}  compare DBAs
used in EPON with alternatives proposed for  GPON. The DBA algorithm we propose is similar
in  principle to  IPACT, proposed  by  Kramer~\etal \cite{KMP02}  for a  regular EPON  and
extended by  Kwong \etal for WDM~\cite{KHA04}.   More complex class of  service aware DBAs
for a  WDM EPON  are proposed  in~\cite{MRM06a, MRM06b, MRC07,  DAM07,MAM09}.  A  paper by
Gliwan \etal discusses DBA for a  coarse WDM GPON~\cite{GKS09}.  Song \etal recognize that
classical EPON  DBAs can lead to  lost capacity in long  reach EPONs and  suggest the ONUs
should  initiate  multiple interleaved  cycles  or  threads~\cite{SKM09}.  We advocate  an
alternative DBA that avoids loss of capacity whatever the propagation times~\cite{AFRR10}.
 
Most performance  evaluations of DBA  algorithms rely on  simulation.  A few  authors have
developed  analytical models,  generally using  techniques developed  for  polling systems.
Park \etal derived closed form formulas for  the mean packet delay in a symmetric PON with
identical ONUs  and negligible propagation times  under the assumption  of Poisson traffic
and  gated service discipline~\cite{PHR05}.   The two-stage  polling system  identified in
that paper has been further analyzed by van der Mei and Resing \cite{ MR08} when ONUs have
heterogeneous load.  To the best of  our knowledge, polling models have not been developed
explicitly  for WDM  PONs. The  capacity  of classical multiserver  polling models  was analysed  by
Fricker and Jaibi~\cite{FJ98}. Borst and van der Mei introduced the notion of server limit
in \cite{BM98}  but did not  evaluate capacity limits.  Down~\cite{Down98}  presents 
several capacity results for multiserver  polling models with server limits but with unlimited service policies. 
Lastly, the study of head  of line processor sharing systems by Brandt and Brandt is relevant but
unfortunately provides no exact capacity results~\cite{BB07}.  

\vspace{2mm}
\noindent
{\em Mean  Field Limits  of Polling Systems}:  To overcome  the complexity of  the polling
systems associated with these networks, the main contribution of the paper consists in the
analysis  of a  large  system limit  where  both the  number  of ONUs  and  the number  of
wavelengths grow indefinitely together.  This scaling, also called {\em mean field limit},
provides explicit formulas and simple numerical algorithms to derive useful approximations
for realistically sized PONs. It is based  on the fact that, when the approximation holds,
the ONUs become statically independent  in the limit, their interaction being represented
through a  deterministic equation.  See~\cite{Sznitman:06}  for a general  presentation of
these methods.   Mean field limits  have been  used for some  time now in  queueing theory
notably to study circuit switched networks, see~\cite{Meleard:94} for example.

In Section~\ref{sec:PON} we  outline our proposed DBAs for EPON and  GPON.  We then recall
relevant  multiserver  polling system  results  and discuss  the  difficulty  of an  exact
analysis   of   the   WDM   PON.    The   large  system   asymptotic   is   developed   in
Section~\ref{sec:largesystem}  where  practically  useful  traffic  capacity  results  are
derived using the mean field assumption.  This property is proved for a particular network
configuration in Section~\ref{sec:convergence}. The validity of this assumption is tested,
in  penultimate Section~\ref{sec:evaluate}  and  used  to evaluate  the  impact of  switch
overhead on PON traffic capacity. Section~\ref{Conc} concludes the paper.

\section{Dynamic bandwidth allocation}\label{sec:PON}

The DBAs we consider for both EPON and GPON relate to a single, classless queue per
user. The objective is to ensure low latency by imposing grant limits per visit and to
ensure users receive adequate throughput by controlling fairness.  

\subsection{EPON DBA}
\noindent
The EPON DBA is governed by the multi-point control protocol (MPCP) based on the exchange
of REPORT and GATE messages between OLT and ONUs, as specified in IEEE 802.3ah. An
extension to WDM is proposed by McGarry \etal~\cite{MRM06a}. ONUs report current contents
of their user queues via REPORT messages while GATE messages are sent downstream by the OLT
to allocate grants for each queue.  In~\cite{AFRR10}, we proposed a novel DBA where the
OLT precisely times the sending of GATE messages and the upstream return of data packets
and REPORTs to maximize traffic capacity. The timing is designed to compensate for
differences in propagation times and to avoid wasting capacity when propagation times are
large. 

The EPON behaves then like a polling system. The ONUs are ``visited'' by each wavelength
following a certain schedule. The visit enables the transmission of upstream data
and a REPORT. 
The grants account for maximum per queue quotas. Each visit incurs a switch overhead time equal to the time needed to send the REPORT plus a physical layer guard time. The
overhead is around 2$\mu$s for a 1Gb/s EPON \cite{SCA09}. The OLT is aware of all
allocations and readily implements constraints on the number of simultaneous transmissions
an ONU can sustain. 

\subsection{GPON DBA}
\noindent
The GPON standards define a particular framework for DBA that allows resource sharing
between several communication services (see the G.984 series of ITU recommendations). We
consider only the packet switched service. Ethernet packets are transmitted within GPON
transmission convergence  (GTC) layer frames of 125$\mu$s fixed duration. Downstream GTC
frames are also used to convey grants from OLT to ONU users and reports on queue contents
are included in upstream GTC frames from ONU to OLT.  Grants in a downstream GTC frame
attribute transmission time slots in a precise upstream GCT frame occurring after a fixed
time offset.  The offset is calculated to allow all ONUs to fulfill their grants
accounting for processing time and propagation delays.  

It is possible to design a DBA in this framework that emulates a polling
system. Wavelength capacity is allocated sequentially to ONUs for data transmission, or
just for reporting if the queues are currently empty. Overhead is incurred on switching
from one ONU to the next, as for the EPON. However, this overhead is rather small
according to the GPON specification, accounting for only some 100$\eta$s \cite{SCA09}. An additional
fixed overhead is incurred at the start of each frame. 

The GPON DBA might alternatively consist in providing every ONU queue a share of every
upstream GTC frame on one or several wavelengths. The OLT would calculate an allocation
for each queue and a schedule of transmissions based on previously received reports and
accounting for ONU transmission capacity. In this way, there is a fixed overhead per
frame, accounting for reports and physical constraints, to be subtracted from the overall
capacity. The remaining capacity would be shared between users according to some fairness
policy. 

\section{Multiserver polling systems}
We introduce notation and discuss known polling system results that can be used to analyse
WDM PONs. 
\subsection{Notation and assumptions}
 In the following wavelengths are  identified with servers  circulating among ONUs.   The polling
system has $N$ ONUs and $L$ wavelengths.   ONU $i$ is equipped with $t_i$ tunable transmitters
and manages  $n_i$ independent user  queues, i.e., $t_i$  servers can be present  at ONU~$i$
simultaneously. When  the DBA emulates  a polling system,  the $j^{th}$ queue of  ONU $i$,
queue $(i,j)$, has a  maximum grant of $d_{ij}$ seconds per visit  by any wavelength.
This  mechanism is known as the limited-gated service policy.  Finally, $\Delta_i$ 
is the switch overhead time associated with ONU $i$. 

We make no assumptions about the traffic in the user queues except that the packet arrival
process is stationary and has a well  defined rate. The traffic intensity of queue $(i,j)$
is $\rho_{ij}$ expressed in  units of the bit rate of one  wavelength. We define $\rho_i =
\sum_j \rho_{ij}$ and $\rho = \sum_i \rho_i$.

A wavelength is assumed to visit ONUs  in a fixed cyclic order, assumed different for each
wavelength. For  \textit{periodic} polling, visits to  ONUs occur  according to a
deterministic schedule.  A wavelength that should visit an ONU with no free transmitter is
assumed to skip  directly to  the next ONU in  the schedule. Alternatively, for
\textit{random} polling, the next ONU  to visit is chosen  at  random  from  the  ONUs
that  currently have  a free  transmitter. 

\subsection{Limited grants, no server limits}
In this section we assume $t_i=L$ for all $i$ -- an ONU can be served by an arbitrary number
$\leq L$ of servers with a limited-gated policy. Polling is either periodic or random.
\begin{proposition}
\label{prop:multiserver}
All queues in all ONUs are stable if and only if, for all $i,j$, 
\begin{equation}
\rho + \frac{\rho_{ij}}{d_{ij}} \sum_{k=1}^N \Delta_k < L.
\label{eq:multistable}
\end{equation}

When not all queues satisfy (\ref{eq:multistable}), there is a subset $\mathcal{S}$ whose
members are saturated while, for $(i,j) \notin \mathcal{S}$, we have  
\begin{equation}
\rho - \sum_{(l,k) \in \mathcal{S}} \rho_{lk} + \frac{\rho_{ij}}{d_{ij}} \left( \sum_{k=1}^N\Delta_k+\sum_{(l,k) \in \mathcal{S}} d_{lk}\right) < L.
\label{eq:multiunstable}
\end{equation}
\end{proposition}

\begin{proof}
Condition (\ref{eq:multistable}) derives from Theorem 3 in \cite{FJ98} or Theorem 2.4 in
\cite{Down98} with appropriate changes in notation. Relaxed conditions in
(\ref{eq:multiunstable}) are deduced from the discussion in Section 4 of \cite{FJ98}. In
fact, the results of \cite{FJ98} and \cite{Down98} must be modified slightly since grant
limits here apply to transmission time and not packets but this is straightforward. 
\end{proof}

We later require the following additional result.
\begin{proposition}
When the system is stable, the mean time $C$ between successive visits of wavelengths to any given ONU is
$C = (\Delta_1+\Delta_2+\cdots+\Delta_N) /(L - \rho)$.
\end{proposition}

\begin{proof}
The  result may  be  deduced  for periodic  polling  from Section  3  in \cite{BM98}.  The
following alternative proof is valid also for random polling.  Since the system is stable,
the mean number  of wavelengths transmitting data is  $\rho$. Thus $L - \rho$  is the mean
number of wavelengths currently busy due to switch overhead. The fraction of time a server
at ONU $i$ experiences a switch overhead  is $\Delta_i/C$. The relation for $C$ is obtained by
summing up these quantities.
\end{proof}

\subsection{Server limits, no grant limits}
The capacity of a multiserver polling system with server limits but unlimited gated service with
periodic or random polling is given by the following.

\begin{proposition}
\label{prop:tunable}
When the WDM EPON with tunable transmitters has unlimited gated service, it is 
stable if $\rho_i < t_i$ for $1\leq i\leq N$ and $\rho_1+\rho_2+\cdots+\rho_N < L$.
\end{proposition}
\begin{proof}
The conditions are clearly necessary since otherwise at least one user queue would be
unstable. To prove sufficiency, we use fluid limit arguments as in \cite{FJ98, Down98}.
Starting from any initial conditions, it is necessary to prove that the overall fluid
content denoted ${\overline{W}}(t)$ decreases with time. Let $\mathcal{F}(t)$ designate
the number of ONU queues that have a fluid backlog at $t$. If $\sum_{i \in \mathcal{F}(t)}
t_i \geq L$, all the service capacity is devoted to fluid queues and
$\dot{\overline{W}}(t) = \sum_{i} \rho_i - L < 0 $  by the second condition. If $\sum_{i \in \mathcal{F}(t)} t_i < L$, all queues in
$\mathcal{F}(t)$ will be served at their maximum rate. We have, therefore,
$\dot{\overline{W}}(t) = \sum_{i \in \mathcal{F}(t)}(\rho_i - t_i)< 0 $. 
\end{proof}
 
Combining server limits and grant limits leads to a system whose traffic capacity turns out
to be intractable, even for a simple toy example with 3 ONUs with a single queue, 2 wavelengths and no
switch overhead \cite{AFRR10}. This motivates the large system approximations discussed in
the next section. 
 
\section{Large system asymptotic}\label{sec:largesystem} 
In this section we derive the traffic capacity of a WDM PON with tunable transmitters
under the assumption that a mean field limit is valid in large systems. 

\subsection{Negligible switch overhead} \label{sec:nooverhead}

\begin{figure}
\centering
\scalebox{0.7}{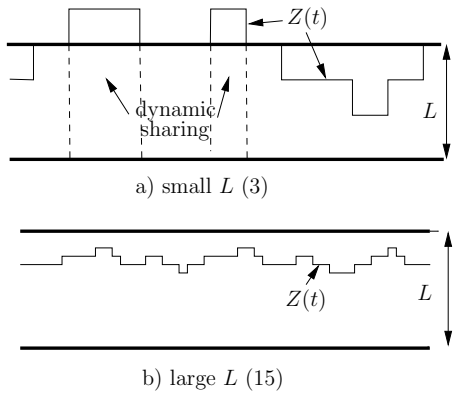}
\caption{Sharing $L$ wavelengths between $Z(t)$ active queues - no switch times}\label{fig:sharing}
\end{figure}
As a preliminary, consider the polling system representing a WDM PON under the assumption
that overhead (for switching between ONUs) is negligibly small.  Assume
the ONUs are each equipped with one tunable transmitter. Let the number of active ONUs at
time $t$ be $Z(t)$, i.e., $Z(t)$ is the number of ONUs with at least one backlogged
queue. When $Z(t) \leq L$, every ONU is served by one wavelength. If $Z(t)>L$,
on the other hand, the PON shares the capacity of the $L$ wavelengths according to the DBA
algorithm.  

Figure \ref{fig:sharing} illustrates the importance of the value of $L$. In the top
figure, $L=3$ and, depending on the ONU loads, it is a frequent occurrence that
$Z(t)>L$. ONUs then share capacity dynamically and it is not easy to predict PON
performance. In the lower figure, $L=15$ and the fluctuations of $Z(t)$ have relatively
lower amplitude. In the depicted realization, every PON is always able to use its tunable
transmitter whenever it is active.  

As $L$ and $N$ increase indefinitely, it is intuitively clear that the fluctuations of
$Z(t)$ become small compared to $L$.  Each ONU then behaves like an independent
polling system with the server capacity of one wavelength. Every ONU queue is stable if
$\ \rho_i < 1$ for $1 \leq i \leq N$ and $ \rho < L$. 

This intuition is confirmed by the propositions of Section~\ref{sec:convergence} in the
case where ONUs have a single queue, service is limited to one packet per cycle, packet
arrivals are Poisson and packet sizes exponential.  A more general proof is, for the
moment,  out of reach. We nevertheless proceed in this and the next sections with the 
reasonable assumption that a ``mean field'' asymptotic is valid as $L$ and $N$ both tend
to infinity. 

It is clear that the above capacity limits can be generalized to ONUs equipped with a
variable number of tunable transmitters. In a large system, each ONU always has access to
its quota of wavelengths whenever it is active. The traffic capacity of ONU $i$ is thus
that of a $t_i$-server polling system: queues are stable if $\rho_i < t_i$ for $1 \leq i
\leq N$ and $ \rho < L$.  


\subsection{Limited service polling system with switch overhead}
 
\begin{figure}
\centering
\scalebox{0.7}{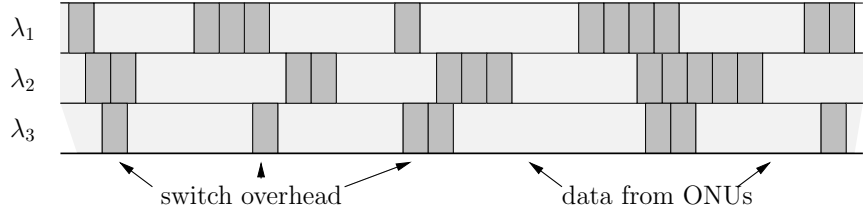}
\caption{Occupation of a 3-wavelength PON by data and overhead}
\label{fig:occupation}
\end{figure}

We now suppose switch overheads are non-negligible: each visit to ONU $i$ consumes
$\Delta_i$ seconds of wavelength capacity. Again, first assume each ONU is equipped with
just one tunable transmitter. Figure \ref{fig:occupation} shows how wavelengths are fully
used in a WDM PON. Each wavelength is either transmitting data or is unavailable because
of the switch overhead. When $\rho<L$, the visit frequency to inactive ONUs is such that
residual capacity $L-\rho$ is entirely consumed by overhead. 

\begin{figure}
\centering
\scalebox{0.7}{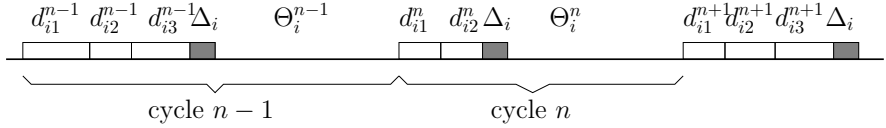}
\caption{ONU local polling system}\label{fig:multipon}
\end{figure}

Figure \ref{fig:multipon}  illustrates the  upstream activity  of an ONU.  This ONU  has 3
queues that  are served cyclically respecting maximum  grant limits on each  visit. In the
figure, queue 3 happens  to be empty in cycle $n$. Each set  of grants terminates with the
switch overhead. After the overhead, the wavelength visits another free ONU (possibly this
one again)  and the  considered ONU sees  a vacation  time ($\geq 0$)  before it  is again
visited by a wavelength. 

When the  system is stable, at  equilibrium the vacation time  of ONU $i$,  i.e., the time
when  no server  is present,  is denoted  by $\Theta_i$.   Note
that  this  random variable depends  \textit{a priori}  on the  activity of  all ONUs.  Under the
mean field  assumption below, this dependence vanishes  in the large system  limit where $N$
and $L$ tend to  infinity.  Let $\theta = \E[\Theta_i]$ and $\delta$ be the switchover
time at equilibrium over all ONUs. First, we
derive a useful general relation. 
\begin{proposition}
At equilibrium, the mean vacation interval $\theta$ of a system with one tunable
transmitter per ONU satisfies
\begin{equation}
\theta = \frac{(N-L)\delta}{L-\rho}. \label{eq:theta}
\end{equation}  
\end{proposition}
\begin{proof}
At equilibrium,  the  number of  ONUs that are not  currently being served is constant
and equal  to $N-L$. Consequently $(N-L)/\theta$  is the rate  at which one of  these ONUs
receives the visit of a server. Similarly, the expected number of wavelengths transmitting
data must  be equal to the  offered traffic $\rho$.  The remainder, $L-\rho$, is  the mean
number of servers in overhead period and consequently $(L-\rho)/\delta$ is the rate at
which servers leave a switch overhead period. The identity of the proposition follows.
\end{proof}
To exploit Equation~(\ref{eq:theta}) we must make the following additional assumption.

\noindent \textit{Mean field assumption}
We assume successive intervals $\Theta^n_i$ are independent and identically distributed
for all $i$ and all $n$ with mean $\theta$. This means the local ONU polling system with
vacation depicted in Figure \ref{fig:multipon} can be analyzed as an independent system,
the impact of the other ONUs being manifested through the value of the mean field
parameter $\theta$. 

We maintain  this assumption is reasonable  when $N$ and $L$  are large and  when, for each
wavelength, the choice  of the next ONU to  visit is made uniformly at  random among those
that are  eligible.  The latter  assumption applies for  both random polling  and periodic
polling in  the large system  limit.  This, roughly speaking,  is because
the  past activity of  any ONU  has negligible  impact on  future wavelength  visits while
available destinations for a wavelength's next visit hardly depend on its own history. The
system is effectively  memoryless and, indeed, we would further  claim the distribution of
the  $\Theta^n_i$  is exponential  in  the  large system  limit.  This property is  not
used for the present capacity results.

To evaluate $\theta$, note that, under the mean field assumption, the system represented
in Figure \ref{fig:multipon} behaves like a classical single server periodic polling
system with  limited gated service and expected overhead per cycle equal to $\Delta_i +
\theta$. By Proposition \ref{prop:multiserver}, its mean cycle time is $(\Delta_i +
\theta)/(1-\rho_i)$.  Straightforward calculations show that the
overall mean overhead $\delta$ defined above can then be expressed as 
\begin{equation}
\delta =\frac{\sum (1-\rho_i)\Delta_i/(\Delta_i + \theta)}{\sum (1-\rho_i)/(\Delta_i + \theta)}.
\label{eq:delta}
\end{equation}
Eliminating $\delta$ in (\ref{eq:theta}) and (\ref{eq:delta})  yields an equation that can
be solved for $\theta$. When the $\Delta_i$ are all equal (to $\delta$), $\theta$ is given
directly by (\ref{eq:theta}). Otherwise, we must resort to numerical evaluation. 

Stability conditions result from those of the individual ONU polling systems where the overhead and vacation are assimilated to the per cycle switch time.

\begin{proposition}
Under the mean field assumption, the WDM PON with one tunable transmitter per ONU,
represented as a multiserver, limited-gated polling system where no more than one server
can simultaneously attend any queue, is stable if  
\begin{equation}
 \rho_i + \frac{\rho_{ij}}{d_{ij}}(\Delta_i +\theta) < 1, \quad 1\leq i\leq N,\quad 1\leq
 j\leq n_i.
\label{eq:conditions}
\end{equation}
where $\theta$ is derived from (\ref{eq:theta}) and (\ref{eq:delta}).
\end{proposition}

The proof follows directly from Proposition \ref{prop:multiserver} on making the appropriate substitutions. When the overhead is the same for each ONU we readily deduce the  following.
\begin{corollary}
When $\Delta_i=\delta$ for all $i$, every queue is stable if 
\begin{equation}
\rho + \max_{ij} \left\{ \frac{\rho_{ij}}{d_{ij}} \frac{\delta(N-\rho)}{1 -\rho_i} \right\} < L.
\label{eq:samedelta}
\end{equation}
\end{corollary}

\subsection{Local stability}
\label{sec:localstability}
Limited gated service is  designed to limit the impact on ONU  performance when some users
saturate  their   share  of   PON  capacity.    Let  the  set   of  saturated   queues  be
$\mathcal{S}$. To derive stability conditions for  the other queues we must simply replace
$\Delta_i$  by $\Delta_i  +  \sum_{j:(i,j)\in\mathcal{S}} d_{ij}$  in (\ref{eq:delta})  to
compute $\theta$ and substitute for $\Delta_i$ and $\theta$ in (\ref{eq:conditions}).

For a system with given traffic distribution, the following algorithm determines the set $\mathcal{S}$ of saturated queues.
\begin{enumerate}
\item $\mathcal{S} =\phi$, the empty set. 
\item find $(i,j) = \argmax_{(i,j) \notin \mathcal{S}} \{\frac{\rho_{ij}}{d_{ij}}\Delta_i\}$. \label{algoloop}
\item if (\ref{eq:conditions}) is satisfied for $(i,j)$ then stop.
\item else,
\begin{itemize}
\item $\Delta_i \leftarrow \Delta_i+d_{ij}$,
\item $\rho  \leftarrow \rho - \rho_{ij}$,
\item $\mathcal{S} \leftarrow \mathcal{S} \cup (i,j)$.
\end{itemize}
\item repeat from step \ref{algoloop}.
\end{enumerate} 

\subsection{Multiple transmitters}
To simplify the presentation, we have so far assumed each ONU has exactly one tunable
transmitter. However, the above results apply immediately with slight modification to a
variable number of transmitters if we make the following assumption. The wavelengths on
terminating one service visit an eligible ONU with probability proportional to its number
of free transmitters. This would result in the large system limit if each wavelength had a
fixed cycle of visits chosen as a random permutation of \textit{transmitters}.  
In condition (\ref{eq:conditions}), it is necessary to replace 1 on the right hand side by $t_i$. 

\subsection{Frame based GPON DBA}

In GPON, a possible DBA is to divide the capacity available in a frame (or several
successive frames) between ONU queues, applying a max-min fair algorithm. The overhead per
frame is fixed and equal to $\sum \Delta_i$. The residual capacity would be shared by the
OLT between active ONU queues. Queue $j$ of ONU $i$ would receive an allocation when
active proportional to its weight  $w_{ij}$.  

Assume the frame duration $f$ and the $\Delta_i$ tend to zero while the ratio $\Delta_i/f$
remains equal to $\delta_i$. The underlying service system is now weighted head of line
processor sharing with a maximum service rate for ONU $i$ equivalent to
$t_i(1-\delta_i)$. 

The capacity of this system appears intractable in general
\cite{BB07}. In the large system limit where the number of queues and the system capacity
increase indefinitely, however, we intuitively expect the service rate constraint to apply
to all ONUs, i.e.,  with high probability, every ONU is served when busy at its maximum
rate. The system in fact behaves as discussed in Section~\ref{sec:nooverhead} and the
stability conditions are  
\begin{eqnarray}
\rho < L(1-\sum \delta_i)\textrm{, and }
\rho_i < t_i(1 - \delta_i).
\end{eqnarray}

\section{Proof of  Mean Field Convergence} \label{sec:convergence}
The above results rely on the assumption that, as $N$ and $L$ get large, the individual
ONU queue sets behave as independent systems. This is a classical mean field assumption
that would ideally need formal justification. In this section we prove the mean field
property under some further simplifying assumptions.  

\noindent \textit{Homogeneous limited server polling system without switch overhead times.} 
We consider an $S_N$-server polling system where each of $N$ queues can be served by only one server at a time. After each service, a server chooses a new, non-empty queue uniformly at random. Moves between queues are instantaneous, there is no switchover time. Customer arrivals are Poisson of intensity $\lambda$ in each queue and service times are exponential with rate $\mu$. The queue load is $\rho=\lambda/\mu$.

For $t\geq  0$,  let $Q_i(t)$  be  the  number  of  customers  in the  $i$th  queue  at  time  $t$  and 
$I_i(t)\in\{0,1\}$ an indicator of the  presence of  a server  at  queue  $i$. Clearly, if $Q_i(t)=0$ then $I_i(t)=0$. The process
\[
(X(t))\stackrel{\text{def.}}{=} ((Q_n(t),I_n(t), 1\leq n\leq N))
\]
is then a Markov process taking values in ${\cal S}=\{x=(x_n,i_n)\in(\N\times\{0,1\})^N:x_n>0 \text{ if } i_n=1\}$.

We consider a ``large'' system where
\[
\lim_{N\to+\infty} \frac{S_N}{N}=s, \text{ with } s\in(0,1).
\]
Note that  the ``natural'' necessary stability conditions are
$N\rho< S_N$ or, in the limit, $\rho<s$. 

First assume the initial occupancies $Q_i(t)$ of each queue are independent and geometrically distributed with parameter $\rho$.   As $N$ grows large, we have
\[
\sum_{n=1}^N \ind{Q_n(0)>0}\sim N \P(Q_1(0)>0)=N\rho< Ns\sim S_N.
\]
Under this regime there is an infinite number of idle servers so that the vector
$(I_n(0))$ is uniquely determined. The following proposition shows that
this holds with high probability on any finite time interval. 
\begin{proposition}
If $\rho<s$ and the initial state $(X(0))=((Q_n(0),I_n(0))$ is such that the variables $Q_n(0)$ are
independent with a geometric distribution with a parameter $\rho$ then, for each $T\geq 0$, there exists $\eps>0$ such that 
\[
\lim_{N\to +\infty}\P\left(\sup_{0\leq s\leq T}  \sum_{n=1}^N \ind{Q_n(t)>0} < S_N-\eps N \right)=1.
\]
\end{proposition}
\begin{proof}
A simple coupling argument is used. Consider $N$ independent $M/M/1$ queues with lengths $L_n(t)$ with
arrival rate $\lambda$ and service rate $\mu$ at equilibrium. For each $t\geq 0$,
$L_n(t)$, $1\leq i\leq N$ are then $N$ independent geometric random variables with parameter
$\rho$. Let 
\[
A_N(t)=\frac{1}{N}\sum_{n=1}^N \ind{L_n(t)>0}.
\]
By the law of large numbers, almost surely,
\[
\lim_{N\to+\infty} A_N(t)=\rho. 
\]
In order to prove a stronger statement, i.e., that this convergence holds not only for a
fixed time but on a whole time interval, one has to work in the Skorokhod space of probability
distributions on real valued
functions on $[0,T]$ with limits on the left and continuous on the right. See
Billingsley~\cite{Billingsley}, for example.   

In the following, $({\cal N}_{x,i})$ denotes an i.i.d. sequence of Poisson
processes with parameter $x$ and ${\cal N}_{x,1}([a,b])$ is the number of points in the
interval $[a,b]$. We first prove that the sequence of
processes $(A_N(t))$ is tight. Since a large variation of $(A_N(t))$ can only be due to
either arrivals or departures, we have, for $\delta$ and $\eta>0$,
\begin{multline*}
\P\left(\sup_{1\leq s,t\leq T: |t-s|\leq \delta} |A_N(t)-A_N(s)|\geq \eta\right)
\\ \leq \P\left(\sup_{1\leq s,t\leq T: |t-s|\leq \delta} \frac{1}{N}\sum_{n=1}^N {\cal N}_{\lambda,n}([s,t])\geq
\frac{\eta}{2}\right)
\\
+ \P\left(\sup_{1\leq s,t\leq T: |t-s|\leq \delta} \frac{1}{N}\sum_{n=1}^N {\cal N}_{\mu,n}([s,t])\geq
\frac{\eta}{2}\right).
\end{multline*}
Using the fact that ${\cal N}_{\lambda,1}+\cdots+{\cal N}_{\lambda,N}$ has the same distribution as
${\cal N}_{N\lambda,1}$, by the law of large numbers for Poisson processes, 
the first term of the right hand side of the above expression becomes arbitrarily small as $N$
grows if  $\delta$ is less than $\eta/2\lambda$. Similarly, the second term
has a similar property. We deduce from this estimate that the sequence of processes
$(A_N(t),0\leq t\leq T)$ is tight for the convergence of the uniform norm and, therefore,
that any of its limiting points is a continuous process. Since the marginals converge to
the constant $\rho$, the only limiting process is the constant process equal to $\rho$.

The sequence $(A_N(t),0\leq t\leq T)$ thus converges {\em as a process } to $(\rho)$. 
In particular, the supremum of $(A_N(t))$ on the time interval $[0,T]$ converges to
$\rho$. Since $S_N\sim s N>\rho N$, for $\eps<s-\rho$, we have  
\[
\lim_{N\to +\infty}\P\left(\sup_{0\leq t\leq T}  A_N(t) < \frac{S_N}{N}-\eps \right)=1.
\]
On the event $\{\sup_{0\leq t\leq T}  A_N(t) < {S_N}/{N}-\eps \}$, there is always an idle
server. In this situation, polling system $(X(t))$ is equivalent to N independent
$M/M/1$ queues so that $(L_n(t),0\leq t\leq T)$ has the same distribution as
$(Q_n(t),0\leq t\leq T)$. 
\end{proof}

The following corollary states the required mean field property given the assumed initial occupancy distribution.
\begin{corollary}
Under the assumptions of the above proposition, on any finite time interval, as $N$ gets
large, the Markov process $(X(t))$ has the same distribution as $N$ independent $M/M/1$
queues. In this limit, the natural condition $\rho<s$ is the stability condition of the system.
\end{corollary}

We conclude this section by a result which states that, starting from any initial state,
the process very rapidly enters a set of states where many servers are idle.
\begin{proposition}
If $\rho<s$ and $x=(q_n,i_n)\in{\cal S}$ is such that $|\{n:q_n>0\}|>S_N$ and
$q_1+\cdots+q_N<CN$ for some constant $C$, the infimum
\[
T_N=\inf\{s, \left|\{n,Q_n(t)>0\}\right|<S_N\},
\]
satisfies
\[
\frac{1}{N}\E_x(T_N)\leq \frac{C}{\mu S_N/N-\lambda}.
\]
\end{proposition}

\begin{proof}
Define, for $t\geq 0$,
\[
Z(t)=Q_1(t)+\cdots+Q_N(t).
\]
Since all servers are busy at time $0$ by assumption, as long as $t<T_N$, $Z(t)-Z(0)$ can be
represented as the difference of two Poisson processes with respective parameters $\lambda
N$ and $\mu S_N$, so that $(Z(t\wedge T_N)+(\mu S_N-N\lambda)(t\wedge T_N))$ is a
super-martingale. In particular
\[
\E(Z(t\wedge T_N)+(\mu S_N-N\lambda)(t\wedge T_N))\leq Z(0)\leq CN.
\]
Hence, $(\mu S_N-N\lambda)\E(t\wedge T_N))\leq  CN$ and the desired inequality
follows on letting $t$ go to infinity.
\end{proof}

\section{Evaluation}\label{sec:evaluate}
We report results of a preliminary evaluation of the accuracy of the mean field approximation and quantify the impact of the switch overhead.

\subsection{Convergence to mean field} 
In the simulations, we consider a realistic PON configuration where each wavelength provides an upstream transmission bit
rate of 1 Gbps. We assume that ONUs have a single queue with infinite capacity and one tunable
transmitter. Users send fixed size 1000 byte packets according to a Poisson process.
The grant size is fixed to $8$ $\mu s$ (transmission time of one packet) and switch overhead is set to 1.2
$\mu s$ (15\% of the grant size).

Figure \ref{fig:onus-switchover} shows the proportion of wavelengths in switch overhead
measured in 1ms intervals for   $N=10,50,100$ and 500, respectively. Overall load is
0.2$N$ and $L/N = 0.5$. The figure shows that, as the number of ONUs increases, the
fluctuations of this proportion decrease. The wavelength visit rate to a given ONU,
determined by this proportion, thus converges to a constant. This is why it is plausible
that individual ONUs tend to become statistically independent.

\begin{figure}[]
\centerline{\epsfig{figure=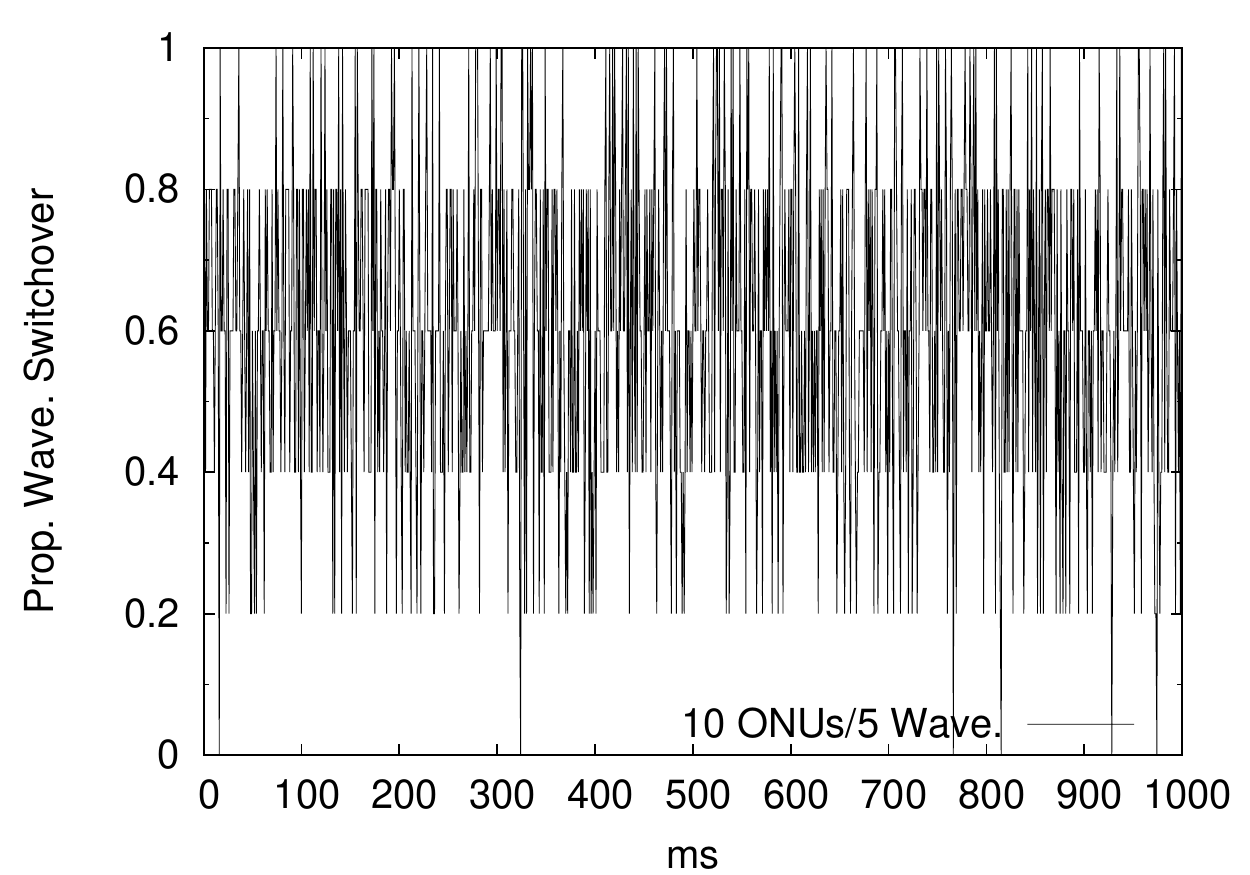,height=1.8in,width=1.8in}
\hspace{-.1in}\epsfig{figure=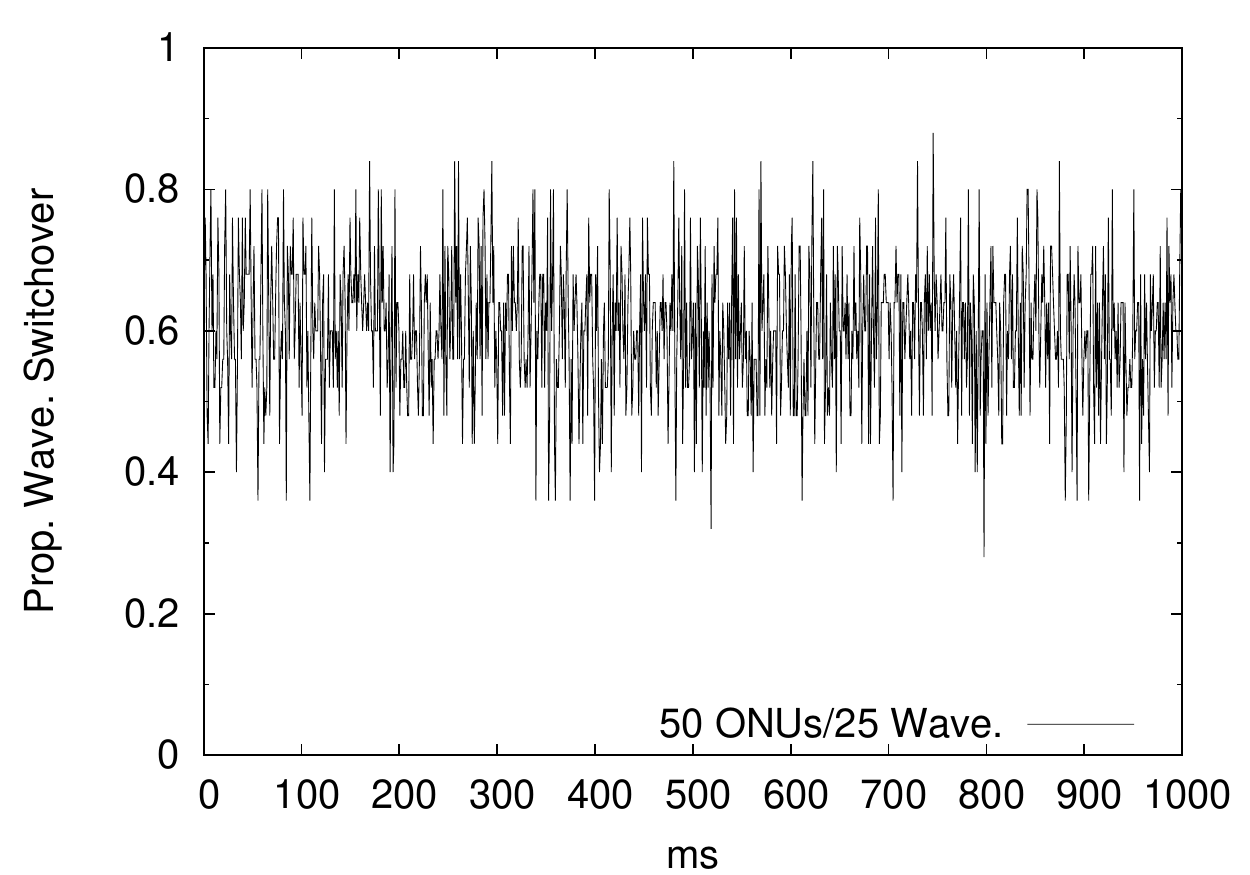,height=1.8in,width=1.8in}}
\centerline{\epsfig{figure=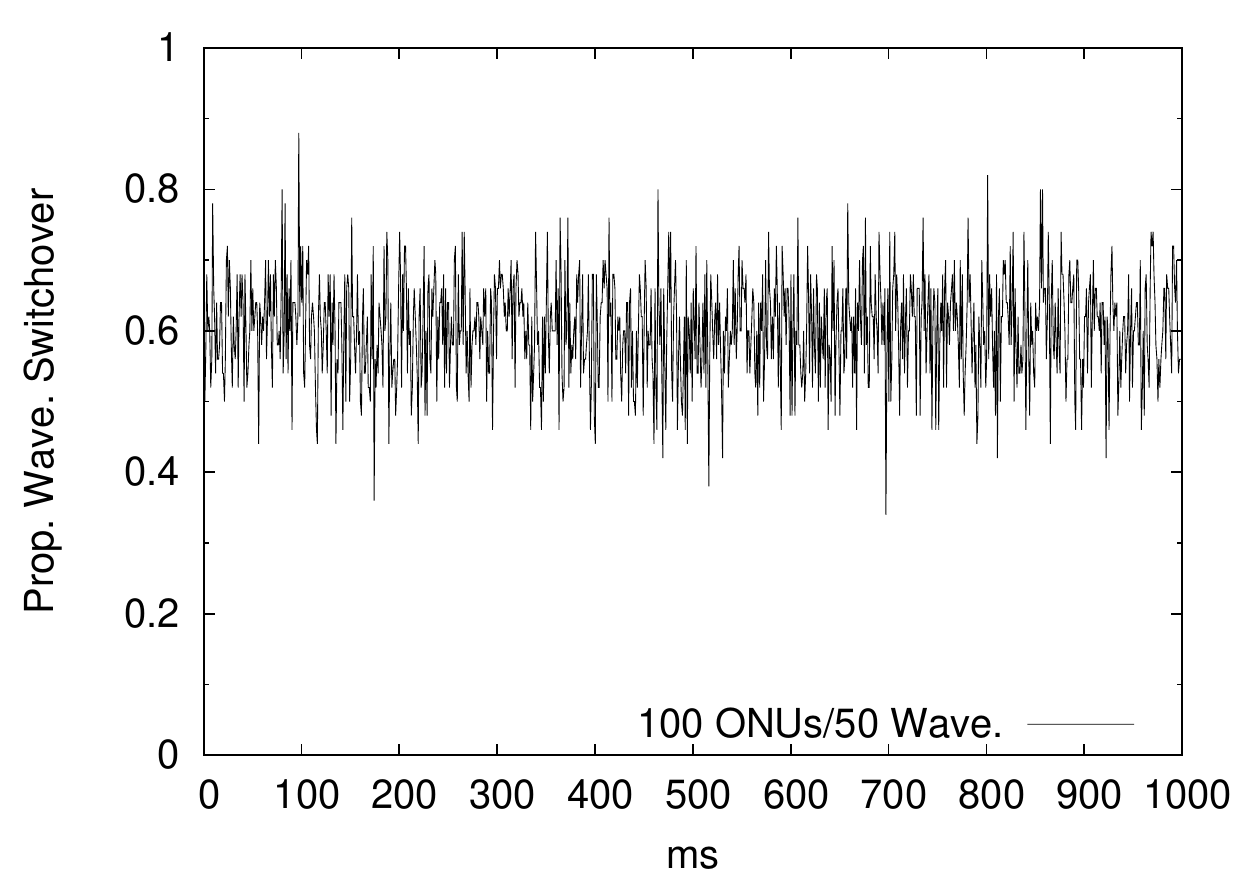,height=1.8in,width=1.8in}
\hspace{-.1in}\epsfig{figure=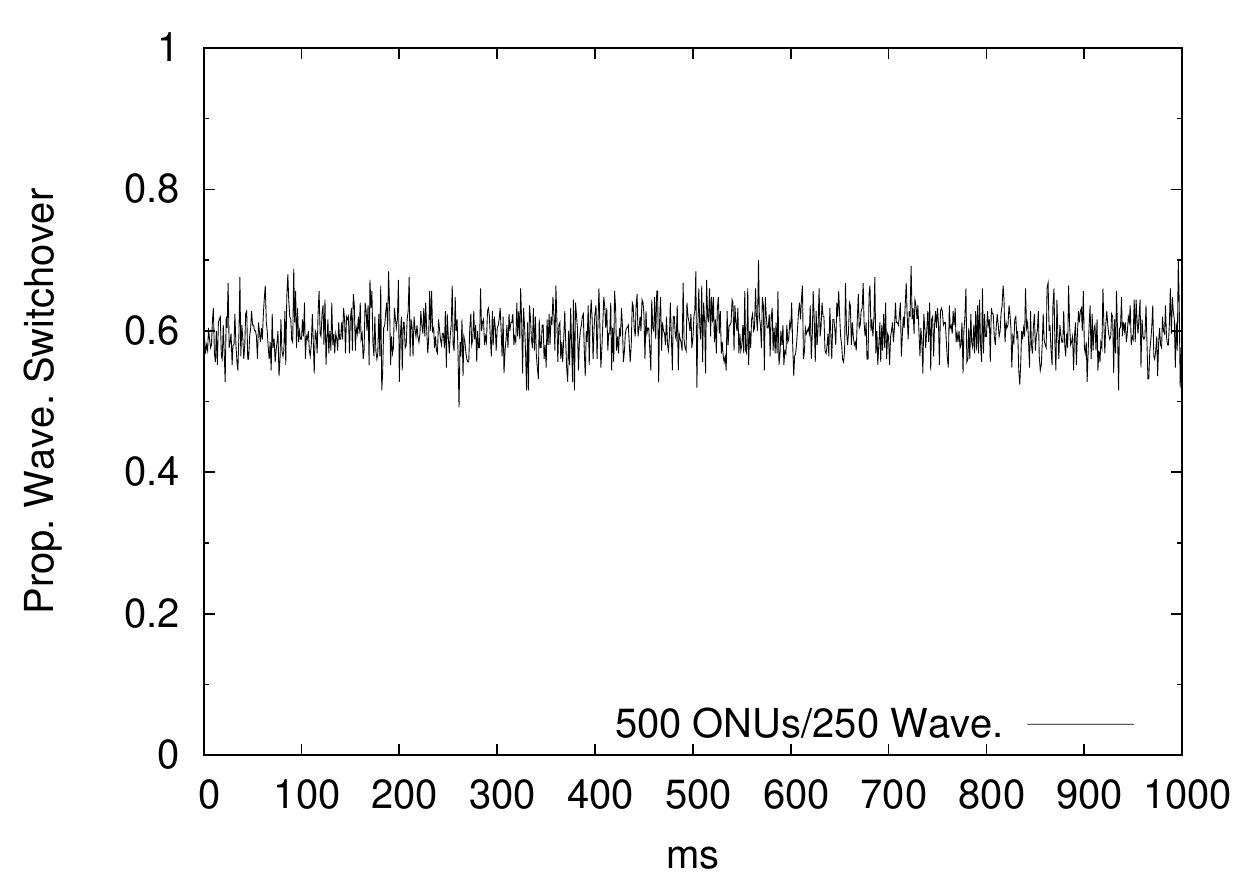,height=1.8in,width=1.8in}}
\caption[]{\label{fig:onus-switchover} Proportion of wavelengths in switch overhead in intervals
of 1 $ms$ for different values of ONUs ($N=10,50,100$ and 500)}
\end{figure}

\subsection{Capacity region}

Figure \ref{fig:stability-balanced} depicts the capacity region of a PON with two classes
of ONUs. Each ONU of the same class has the same load and there is an equal number of ONUs
per class.  $L/N=1/2$. The region to the left and below the plots corresponds to loads for
which the ONU queues are stable.  Units are normalized loads with respect to the total
number of ONUs.  

Figure  \ref{fig:stability-balanced} compares the capacity region determined by
(\ref{eq:samedelta}) with simulation results for $N=6$ and $N=20$. The accuracy of the
results is good for $N=6$ and excellent for $N=20$. Discrepancies for $N=6$ with
unbalanced traffic are not surprising as the number of ONUs is then effectively only
3. Similar precision is observed in Figure \ref{fig:stability-unbalanced} where class 2
has three times as many ONUs as class 1.  Note that PONs are currently designed for around
100 ONUs so that the accuracy of the mean field approximation is largely sufficient.

\begin{figure}[]
\centerline{\epsfig{figure=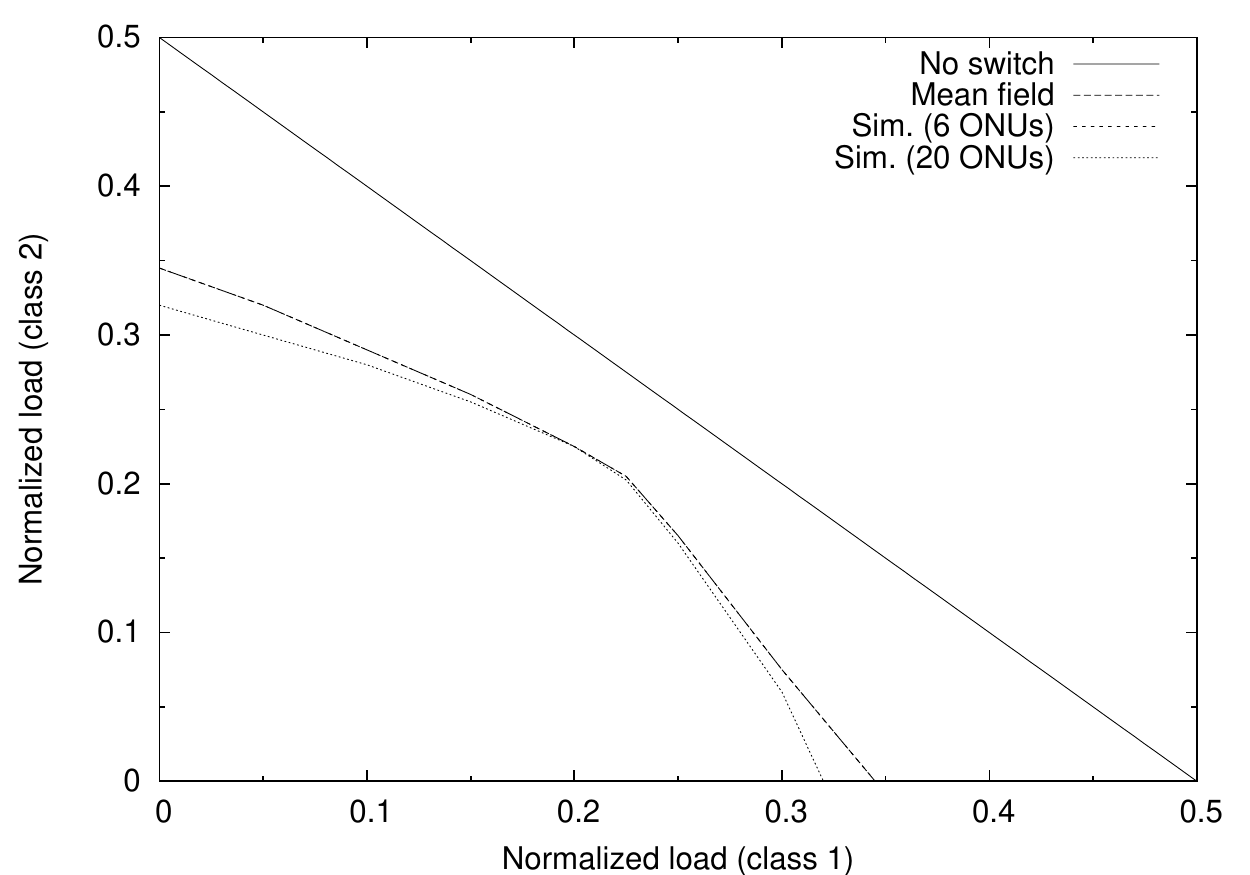,height=2.5in,width=2.5in}}
\caption[]{\label{fig:stability-balanced} Stability region for two classes of ONUs with balanced
number of ONUs}
\end{figure}

\begin{figure}[]
\centerline{\epsfig{figure=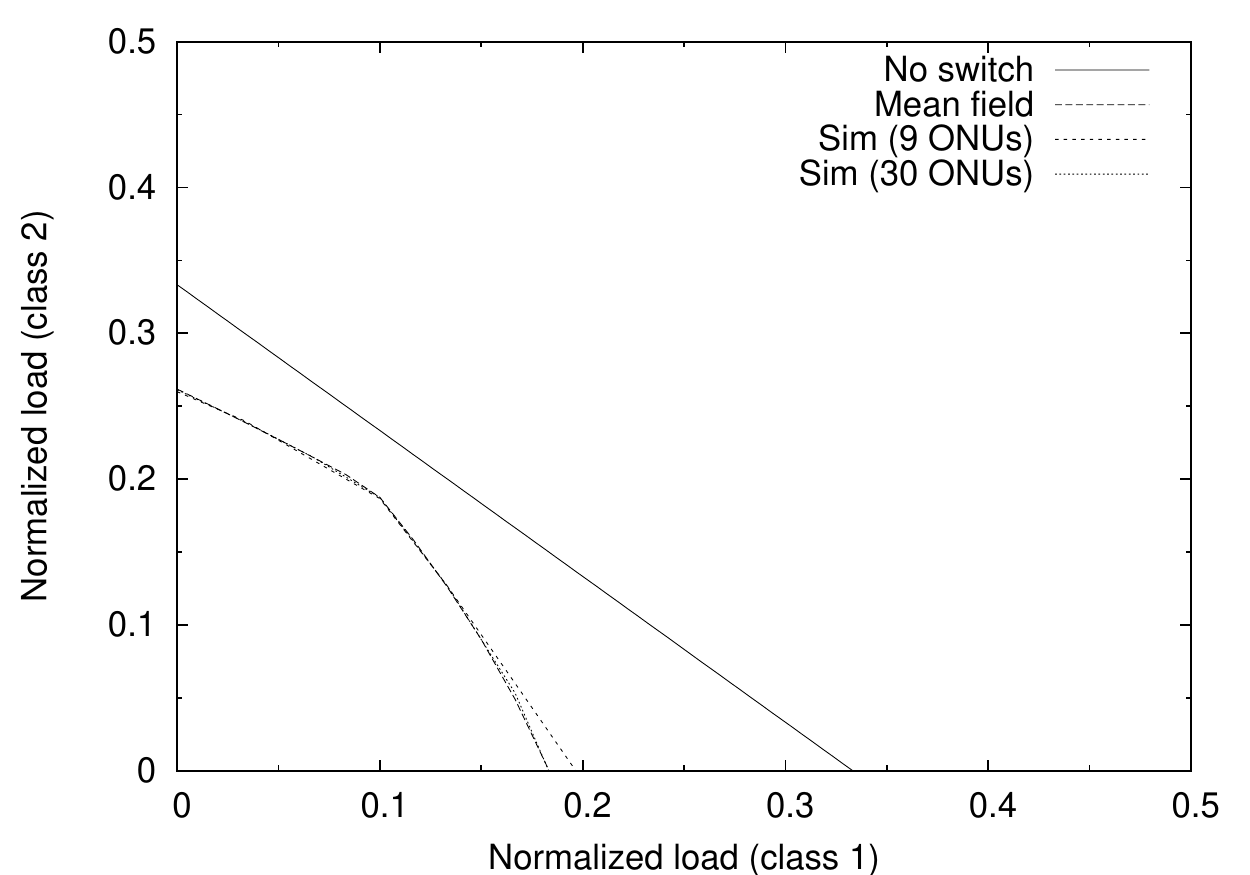,height=2.5in,width=2.5in}}
\caption[]{\label{fig:stability-unbalanced} Stability region for two classes of ONUs with
unbalanced number of ONUs}
\end{figure}

\subsection{Impact of overhead}
Figure \ref{fig:stability-overhead} illustrates the impact of the overhead as a percentage
of the grant limit given by the mean field approximation. For these results, $L/N = 0.4$
and 40\% of ONUs are in class 1. The graphs show that the impact of overhead is
significant.  It is worthwhile increasing the grant size as long as the worst case latency
remains sufficiently small. Note that when some queues are overloaded, their grant is
effectively added to the overhead seen by the other queues, as discussed in
Section~\ref{sec:localstability}.   

\begin{figure}[]
\centerline{\epsfig{figure=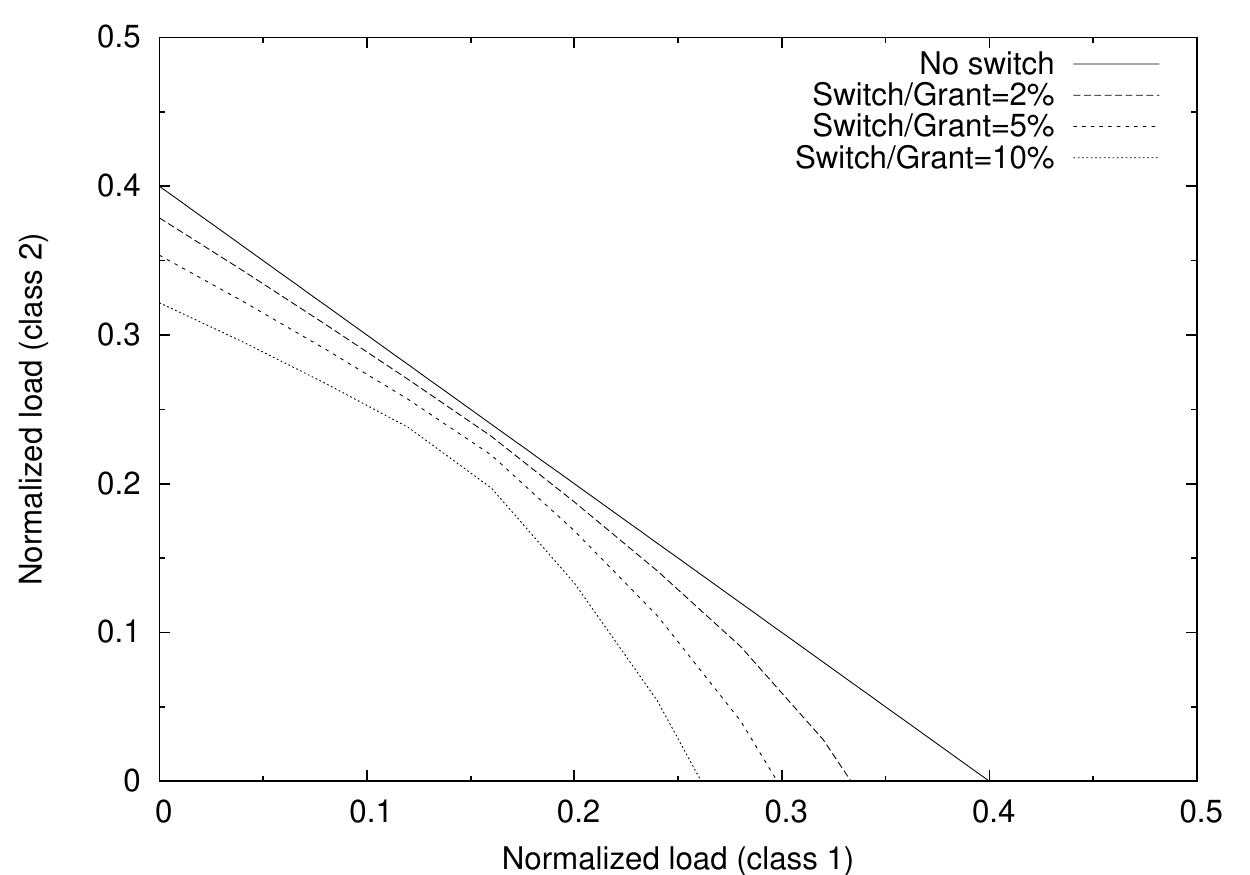,height=2.5in,width=2.5in}}
\caption[]{\label{fig:stability-overhead}  Stability region for two classes of ONUs with
unbalanced number of ONUs for different ratios switchover/Grant Size}
\end{figure}

\section{Conclusion}\label{Conc}
The design of dynamic bandwidth allocation algorithms for WDM PONs brings challenging
performance problems, even when we suppose the DBA does not attempt to realize intra-ONU service differentiation.
The main performance  criterion for these systems is the maximal  traffic capacity they can
achieve, delays being negligibly small until loads are very close to this capacity.  

The natural stochastic models to  be investigated in this domain are multiserver
polling models with both server limits and grant limits. They appear to
be not tractable  in general. For this  reason, we have explored  a large system
approximation based on a mean  field assumption.  Under this assumption,  the behaviour of a
given ONU can be  described as an isolated queueing system independent  of the other ONUs.
The  traffic capacity  analysis for  this limiting  representation of  the PON  notably reveals
the non-trivial impact of  switch overhead. Simulations reveal excellent
accuracy  even  for  systems  much  smaller  than that of real  PONs.  

The  mathematical
justification of  the mean field  approximation has been  provided in a simple  case where
switch times are null and traffic is Markovian.
In future  work we aim to extend  this analysis  to the  realistic case  of general  traffic and
non-zero switch  overhead.   We also intend to use the results derived here to perform a more thorough analysis of DBA performance for EPON and GPON. 

Note finally that the derived
large system approximations  can be applied in other contexts where the natural model is a multiserver  polling
system with both server and grant limits.

\end{document}

%% file: ponpic.pdf_t
\begin{picture}(0,0)%
\includegraphics{ponpic.pdf}%
\end{picture}%
\setlength{\unitlength}{3947sp}%
\begingroup\makeatletter\ifx\SetFigFont\undefined%
\gdef\SetFigFont#1#2#3#4#5{%
  \reset@font\fontsize{#1}{#2pt}%
  \fontfamily{#3}\fontseries{#4}\fontshape{#5}%
  \selectfont}%
\fi\endgroup%
\begin{picture}(6679,2652)(4636,-4351)
\put(10516,-2956){\makebox(0,0)[lb]{\smash{{\SetFigFont{14}{16.8}{\familydefault}{\mddefault}{\updefault}{\color[rgb]{0,0,0}Internet}%
}}}}
\put(4801,-2086){\makebox(0,0)[lb]{\smash{{\SetFigFont{14}{16.8}{\familydefault}{\mddefault}{\updefault}{\color[rgb]{0,0,0}end}%
}}}}
\put(4651,-2311){\makebox(0,0)[lb]{\smash{{\SetFigFont{14}{16.8}{\familydefault}{\mddefault}{\updefault}{\color[rgb]{0,0,0}users}%
}}}}
\put(4801,-3736){\makebox(0,0)[lb]{\smash{{\SetFigFont{14}{16.8}{\familydefault}{\mddefault}{\updefault}{\color[rgb]{0,0,0}end}%
}}}}
\put(4651,-3961){\makebox(0,0)[lb]{\smash{{\SetFigFont{14}{16.8}{\familydefault}{\mddefault}{\updefault}{\color[rgb]{0,0,0}users}%
}}}}
\put(5626,-4336){\makebox(0,0)[lb]{\smash{{\SetFigFont{14}{16.8}{\familydefault}{\mddefault}{\updefault}{\color[rgb]{0,0,0}ONU}%
}}}}
\put(5626,-2611){\makebox(0,0)[lb]{\smash{{\SetFigFont{14}{16.8}{\familydefault}{\mddefault}{\updefault}{\color[rgb]{0,0,0}ONU}%
}}}}
\put(7351,-2161){\makebox(0,0)[lb]{\smash{{\SetFigFont{14}{16.8}{\familydefault}{\mddefault}{\updefault}{\color[rgb]{0,0,0}passive}%
}}}}
\put(7351,-2446){\makebox(0,0)[lb]{\smash{{\SetFigFont{14}{16.8}{\familydefault}{\mddefault}{\updefault}{\color[rgb]{0,0,0}splitter}%
}}}}
\put(7471,-3901){\makebox(0,0)[lb]{\smash{{\SetFigFont{14}{16.8}{\familydefault}{\mddefault}{\updefault}{\color[rgb]{0,0,0}fibre}%
}}}}
\end{picture}%

%% file: sharing.pdf_t
\begin{picture}(0,0)%
\includegraphics{sharing.pdf}%
\end{picture}%
\setlength{\unitlength}{3947sp}%
\begingroup\makeatletter\ifx\SetFigFont\undefined%
\gdef\SetFigFont#1#2#3#4#5{%
  \reset@font\fontsize{#1}{#2pt}%
  \fontfamily{#3}\fontseries{#4}\fontshape{#5}%
  \selectfont}%
\fi\endgroup%
\begin{picture}(4035,3501)(388,-3079)
\put(3121,239){\makebox(0,0)[b]{\smash{{\SetFigFont{12}{14.4}{\familydefault}{\mddefault}{\updefault}{\color[rgb]{0,0,0}$Z(t)$}%
}}}}
\put(1576,-1276){\makebox(0,0)[lb]{\smash{{\SetFigFont{12}{14.4}{\familydefault}{\mddefault}{\updefault}{\color[rgb]{0,0,0}a)  small $L$ (3)}%
}}}}
\put(3106,-2296){\makebox(0,0)[b]{\smash{{\SetFigFont{12}{14.4}{\familydefault}{\mddefault}{\updefault}{\color[rgb]{0,0,0}$Z(t)$}%
}}}}
\put(1621,-766){\makebox(0,0)[lb]{\smash{{\SetFigFont{12}{14.4}{\familydefault}{\mddefault}{\updefault}{\color[rgb]{0,0,0}sharing}%
}}}}
\put(1576,-586){\makebox(0,0)[lb]{\smash{{\SetFigFont{12}{14.4}{\familydefault}{\mddefault}{\updefault}{\color[rgb]{0,0,0}dynamic}%
}}}}
\put(1636,-3001){\makebox(0,0)[lb]{\smash{{\SetFigFont{12}{14.4}{\familydefault}{\mddefault}{\updefault}{\color[rgb]{0,0,0}b)  large $L$ (15)}%
}}}}
\put(4216,-601){\makebox(0,0)[b]{\smash{{\SetFigFont{12}{14.4}{\familydefault}{\mddefault}{\updefault}{\color[rgb]{0,0,0}$L$}%
}}}}
\put(4171,-2236){\makebox(0,0)[b]{\smash{{\SetFigFont{12}{14.4}{\familydefault}{\mddefault}{\updefault}{\color[rgb]{0,0,0}$L$}%
}}}}
\end{picture}%

%% file: occupation.pdf_t
\begin{picture}(0,0)%
\includegraphics{occupation.pdf}%
\end{picture}%
\setlength{\unitlength}{3947sp}%
\begingroup\makeatletter\ifx\SetFigFont\undefined%
\gdef\SetFigFont#1#2#3#4#5{%
  \reset@font\fontsize{#1}{#2pt}%
  \fontfamily{#3}\fontseries{#4}\fontshape{#5}%
  \selectfont}%
\fi\endgroup%
\begin{picture}(7687,1827)(811,-2026)
\put(2176,-2011){\makebox(0,0)[lb]{\smash{{\SetFigFont{14}{16.8}{\familydefault}{\mddefault}{\updefault}{\color[rgb]{0,0,0}switch overhead}%
}}}}
\put(826,-511){\makebox(0,0)[lb]{\smash{{\SetFigFont{14}{16.8}{\familydefault}{\mddefault}{\updefault}{\color[rgb]{0,0,0}$\lambda_1$}%
}}}}
\put(826,-961){\makebox(0,0)[lb]{\smash{{\SetFigFont{14}{16.8}{\familydefault}{\mddefault}{\updefault}{\color[rgb]{0,0,0}$\lambda_2$}%
}}}}
\put(826,-1411){\makebox(0,0)[lb]{\smash{{\SetFigFont{14}{16.8}{\familydefault}{\mddefault}{\updefault}{\color[rgb]{0,0,0}$\lambda_3$}%
}}}}
\put(5776,-2011){\makebox(0,0)[lb]{\smash{{\SetFigFont{14}{16.8}{\familydefault}{\mddefault}{\updefault}{\color[rgb]{0,0,0}data from ONUs}%
}}}}
\end{picture}%

%% file: multiPON.pdf_t
\begin{picture}(0,0)%
\includegraphics{multiPON.pdf}%
\end{picture}%
\setlength{\unitlength}{3947sp}%
\begingroup\makeatletter\ifx\SetFigFont\undefined%
\gdef\SetFigFont#1#2#3#4#5{%
  \reset@font\fontsize{#1}{#2pt}%
  \fontfamily{#3}\fontseries{#4}\fontshape{#5}%
  \selectfont}%
\fi\endgroup%
\begin{picture}(7919,1107)(1404,-1123)
\put(3826,-211){\makebox(0,0)[lb]{\smash{{\SetFigFont{14}{16.8}{\rmdefault}{\mddefault}{\updefault}{\color[rgb]{0,0,0}$\Theta_i^{n-1}$}%
}}}}
\put(4951,-211){\makebox(0,0)[lb]{\smash{{\SetFigFont{14}{16.8}{\familydefault}{\mddefault}{\updefault}{\color[rgb]{0,0,0}$d_{i1}^{n}$}%
}}}}
\put(8776,-211){\makebox(0,0)[lb]{\smash{{\SetFigFont{14}{16.8}{\familydefault}{\mddefault}{\updefault}{\color[rgb]{0,0,0}$\Delta_i$}%
}}}}
\put(7501,-211){\makebox(0,0)[lb]{\smash{{\SetFigFont{14}{16.8}{\familydefault}{\mddefault}{\updefault}{\color[rgb]{0,0,0}$d_{i1}^{n+1}$}%
}}}}
\put(2701,-1036){\makebox(0,0)[lb]{\smash{{\SetFigFont{14}{16.8}{\rmdefault}{\mddefault}{\updefault}{\color[rgb]{0,0,0}cycle $n-1$}%
}}}}
\put(5776,-1036){\makebox(0,0)[lb]{\smash{{\SetFigFont{14}{16.8}{\rmdefault}{\mddefault}{\updefault}{\color[rgb]{0,0,0}cycle $n$}%
}}}}
\put(2176,-211){\makebox(0,0)[lb]{\smash{{\SetFigFont{14}{16.8}{\rmdefault}{\mddefault}{\updefault}{\color[rgb]{0,0,0}$d_{i2}^{n-1}$}%
}}}}
\put(3076,-211){\makebox(0,0)[lb]{\smash{{\SetFigFont{14}{16.8}{\rmdefault}{\mddefault}{\updefault}{\color[rgb]{0,0,0}$\Delta_i$}%
}}}}
\put(5401,-211){\makebox(0,0)[lb]{\smash{{\SetFigFont{14}{16.8}{\rmdefault}{\mddefault}{\updefault}{\color[rgb]{0,0,0}$d_{i2}^{n}$}%
}}}}
\put(5701,-211){\makebox(0,0)[lb]{\smash{{\SetFigFont{14}{16.8}{\rmdefault}{\mddefault}{\updefault}{\color[rgb]{0,0,0}$\Delta_i$}%
}}}}
\put(6301,-211){\makebox(0,0)[lb]{\smash{{\SetFigFont{14}{16.8}{\rmdefault}{\mddefault}{\updefault}{\color[rgb]{0,0,0}$\Theta_i^{n}$}%
}}}}
\put(7876,-211){\makebox(0,0)[lb]{\smash{{\SetFigFont{14}{16.8}{\rmdefault}{\mddefault}{\updefault}{\color[rgb]{0,0,0}$d_{i2}^{n+1}$}%
}}}}
\put(8326,-211){\makebox(0,0)[lb]{\smash{{\SetFigFont{14}{16.8}{\rmdefault}{\mddefault}{\updefault}{\color[rgb]{0,0,0}$d_{i3}^{n+1}$}%
}}}}
\put(1651,-211){\makebox(0,0)[lb]{\smash{{\SetFigFont{14}{16.8}{\rmdefault}{\mddefault}{\updefault}{\color[rgb]{0,0,0}$d_{i1}^{n-1}$}%
}}}}
\put(2701,-211){\makebox(0,0)[lb]{\smash{{\SetFigFont{14}{16.8}{\rmdefault}{\mddefault}{\updefault}{\color[rgb]{0,0,0}$d_{i3}^{n-1}$}%
}}}}
\end{picture}%